\documentclass[conference]{IEEEtran}
\IEEEoverridecommandlockouts
\usepackage{cite}
\usepackage{amsmath,amssymb,amsfonts,amsthm}
\usepackage{algorithmic}
\usepackage{graphicx}
\usepackage{textcomp}
\usepackage{xcolor}
\usepackage{flushend}
\usepackage{caption}
\usepackage{subcaption}
\usepackage{graphicx}
\usepackage{mathabx}
\usepackage{mathtools}

\newtheorem{theorem}{Theorem}

\newtheorem{remark}{Remark}
\newtheorem{definition}{Definition}
\newtheorem{prop}{Proposition}

\DeclareMathOperator*{\argmin}{arg\,min}
\newcommand{\paths}{\mathsf{E}}
\newcommand{\group}{\mathcal{G}}
\newcommand{\capacity}{\mathsf{C}}
\newcommand{\edgedisjoint}{\mathsf{E}}

\def\BibTeX{{\rm B\kern-.05em{\sc i\kern-.025em b}\kern-.08em
		T\kern-.1667em\lower.7ex\hbox{E}\kern-.125emX}}
\begin{document}
	
	\title{Two-Level Priority Coding for Resilience to Arbitrary Blockage Patterns
	\thanks{The research at UCLA was supported in part by the U.S. National Science Foundation (NSF) under grant no. \mbox{ECCS-2229560}. The research at UMN was supported in part by the NSF under grant no. \mbox{CCF-2045237}. This material is also based upon work supported by the NSF under grant no. \mbox{CNS-2146838} and is supported in part by funds from federal agency and industry partners as specified in the Resilient \& Intelligent NextG Systems (RINGS) program.}
    }

\author{%
 \IEEEauthorblockN{Mine Gokce Dogan\IEEEauthorrefmark{1},
 Abhiram Kadiyala\IEEEauthorrefmark{1},
                     Jaimin Shah\IEEEauthorrefmark{2},
                     Martina Cardone\IEEEauthorrefmark{2},
                     Christina Fragouli\IEEEauthorrefmark{1}
                     }
   \IEEEauthorblockA{\IEEEauthorrefmark{1}%
                     University of California, Los Angeles
                     Los Angeles, CA 90095,
                     \{minedogan96, akadiyal, christina.fragouli\}@ucla.edu}
    \IEEEauthorblockA{\IEEEauthorrefmark{2}%
                     University of Minnesota,
                     Minneapolis, MN 55455,
                     \{shah0732, mcardone\}@umn.edu}
 }

	\maketitle
	
\begin{abstract}
Ultra-reliable low-latency communication is essential in \hbox{mission-critical} settings, including military applications, where persistent and asymmetric link blockages caused by mobility, jamming, or adversarial attacks can disrupt \hbox{delay-sensitive} transmissions. This paper addresses this challenge by deploying a multilevel diversity coding (MDC) scheme that controls the received information, offers distinct reliability guarantees based on the priority of data streams, and maintains low design and operational complexity as the number of network paths increases. 
For two priority levels over three \hbox{edge-disjoint} paths, the complete capacity region is characterized, showing that superposition coding achieves the region in general, whereas network coding is required only in a specific corner case. 
Moreover, sufficient conditions under which a simple superposition coding scheme achieves the capacity for an arbitrary number of paths are identified.
To prove these results and provide a unified analytical framework, the problem of designing \hbox{high-performing} MDC schemes is shown to be equivalent to the problem of designing \hbox{high-performing} encoding
schemes over a class of broadcast networks, referred to as combination networks in the literature. 
\end{abstract}

	
	\section{Introduction}
    Ultra-reliable low-latency communication (URLLC) is critical for mission-critical applications, including military systems. However, URLLC is highly susceptible to link blockages, which may result from mobility, environmental factors, jamming, or adversarial attacks~\cite{Jain,Wang17}. Such blockages can cause sudden and persistent outages that severely disrupt \hbox{delay-sensitive} transmissions. Since URLLC applications typically operate over short transmission durations, a blocked link often remains unavailable throughout the transmission window. Moreover, blockage probabilities can vary significantly and are often highly \emph{asymmetric} across different links~\cite{li2008impact,Rangan,thornburg2016performance}. These considerations necessitate the design of resilient transmission mechanisms that can proactively anticipate and mitigate the impact of link blockages.

	Multilevel Diversity Coding (MDC)~\cite{Roche97,Yeung99,Mohajer2008} offers a powerful \textit{proactive} framework to address this challenge. It offers resilience in advance without requiring prior knowledge of the blockages, and allows us to control the received information. Moreover, by encoding \emph{prioritized} information streams with different levels of redundancy, MDC allows the system to meet varying Quality-of-Service (QoS) requirements while enabling a graceful performance degradation. In particular, it ensures that higher-priority data can still be recovered even when only a subset of the network paths remains available.
    
    Despite its advantages, classical formulations of asymmetric MDC~\cite{Mohajer2008} require accounting for all possible blockage patterns (i.e.,
    $2^\paths-1$ patterns for $\paths$ \hbox{edge-disjoint} network paths), resulting in a complexity that exponentially grows as the number of paths increases. Our recent work~\cite{MineISIT2025} addressed this challenge by dividing the blockage patterns into different groups\footnote{The number of groups is determined by the number of priority levels.} instead of analyzing each pattern individually. We showed that our proposed scheme achieves capacity (i.e., it is optimal) for two priority levels when all the blockage patterns that may occur in the network are considered in the design. In this paper, we take a step further and develop an MDC scheme that allows for selective inclusion of blockage patterns, which offers additional design flexibility. This scenario is motivated by practical settings where certain blockage patterns may never occur due to link protection or may occur with such low probability that they can be ignored for design purposes. Our results show that, depending on how the groups are structured, selectively including patterns in the design can improve the achievable rates. As a first step, in this paper we focus on the practically relevant case of two priority levels across information streams, such as high-priority and low-priority data. This setting may be particularly important in military applications, where high-priority streams may carry \hbox{command-and-control signals}, situational awareness updates, or threat alerts that require high reliability guarantees even under degraded conditions. In contrast, lower-priority data may carry background telemetry or non-urgent sensor logs, which can tolerate partial losses.

 Our main contributions are summarized as follows.
    \begin{itemize}
    \item We provide a unified analytical framework by reducing the problem of designing \hbox{high-performing} MDC schemes to the problem of designing \hbox{high-performing} encoding
schemes over a class of broadcast networks, referred to as  combination networks in the literature~\cite{ngai2004network}. 
    
    \item We characterize the complete capacity region (i.e., maximum rates of communication) of the proposed MDC scheme for the case of two priority levels and three edge-disjoint paths. We derive interpretable expressions for the capacity region and propose simple coding schemes that can achieve this region. In particular, we show that superposition coding is optimal in general, while network coding is required only in a specific corner case. 
    \item We identify sufficient conditions under which a simple superposition coding scheme achieves the capacity region for two priority levels and an arbitrary number of network edge-disjoint paths.
     \end{itemize}

	
	\subsection{Related Work}
    	Existing works in the literature provide resilience against blockages through either {\em reactive} or {\em proactive} strategies. Reactive methods~\cite{Jeong,Abari,10024832,MineMilcom,10279412,10439990,le2020overview} respond to blockages after they occur using interleaving and feedback mechanisms, but the resulting delay often makes them unsuitable for URLLC. 
        In contrast, proactive mechanisms attempt to mitigate communication impairments or blockages in advance~\cite{Minepassive,MineISIT22,Mine23Arxiv,SurMobicom,Haider,Wei,Va,Ferreira23}, reducing latency, but they either lack control over the specific information received due to the absence of coding, or they process multiple data streams within a single codeword without accommodating their distinct QoS requirements. 
	
	To address the aforementioned gaps, our prior work~\cite{MineMILCOM23,MineGlobecom,DoganICC24,Mine_journal} proposed MDC schemes to control the received information and offer distinct reliability guarantees based on the priority of data streams. However, these works either overlook asymmetric blockage probabilities or lack guarantees of optimality. In~\cite{Mohajer2008}, the capacity region of an MDC scheme was derived for up to three paths under asymmetric blockages, but the exhaustive analysis leads to complex expressions and high implementation complexity. Our recent work~\cite{MineISIT2025} proposed a grouping strategy that partitions the blockage patterns into distinct groups to reduce design complexity. However, the optimality result in~\cite{MineISIT2025} holds when all the blockage patterns are taken into account. Differently, in this paper we develop an MDC scheme that allows us to select a subset of the blockage patterns, which offers additional design flexibility. Our results show that this flexibility can lead to improved rates depending on the group structures.
	
    \subsection{Paper Organization}
    In Section~\ref{sec:Background}, we provide background on the network model, link blockage characteristics, erasure codes, and asymmetric MDC. In Section~\ref{sec:Problem_Formulation}, we formally define our problem and introduce our MDC scheme. In Section~\ref{sec:reduction}, we establish a connection between our MDC setting and combination networks. In Section~\ref{sec:3_paths}, we present our main technical contributions. Finally, in Section~\ref{sec:conclusion}, we summarize our findings and outline our future research directions.

\section{System Model and Background }\label{sec:Background}
\noindent     \textbf{Notation.}   
    $[a\!:\!b] \coloneqq \{a,a\!+\!1,\dots,b\}$ for integers $a<b$; $|\mathcal{S}|$ denotes the cardinality of a set $\mathcal{S}$;  
    $\mathbf{0}$ is an all-zero vector.

We consider a network that consists of a source (node $0$), a destination (node $N+1$), and $N$ intermediate relays. To facilitate an information-theoretic analysis, we model the network as a directed graph, where edges represent \hbox{point-to-point} communication links with associated capacities. This modeling captures the essential structure of the network while abstracting away physical layer constraints. The resulting network consists of $\paths$ edge-disjoint paths between the source and the destination, each assumed to have capacity equal to $\capacity$. 
The source can transmit $\paths$ packets simultaneously -- one per path -- and the destination can receive them concurrently. Differently, at each point in time, each relay can receive from and transmit to at most one node\footnote{Our results extend to cases where relays have multiple beams.}.

\noindent \textbf{Link Blockages.}
In mission-critical applications such as military communications, links may experience random and potentially adversarial blockages. This can lead to unavailability of paths over the entire transmission window. 
Moreover, blockage probabilities can be highly \emph{asymmetric} across different links~\cite{Rangan,thornburg2016performance}.
For the purpose of designing resilient transmission mechanisms, link blockage probabilities can be accurately estimated in advance by suitably modeling the blocker arrival process~\cite{Jain, Wang17,ahmed2022blockage}.
We build on the existence of such models and assume that these probabilities are correctly estimated. We also assume zero-error capacity channels, where packets are either received without error or entirely lost.

\noindent \textbf{Erasure Codes.}
Erasure codes improve reliability by tolerating a fixed number of packet losses. In particular, an $(n,k)$ erasure code encodes $k$ information packets into $n>k$ coded packets, allowing recovery from any subset of $k$ packets. Thus, it can tolerate up to $n-k$ erasures. However, it does not offer a graceful performance degradation: the information rate drops to zero if fewer than $k$ packets are received; and it remains fixed at $k/n$ if at least $k$ packets are received, regardless of how many packets are received beyond $k$.


\noindent\textbf{Asymmetric MDC.}
We focus on \emph{asymmetric} MDC~\cite{Mohajer2008} -- a generalization of the symmetric variant~\cite{Roche97,Yeung99}, which is well-suited for networks with \emph{unequal} link blockage probabilities. The goal is to encode the source sequences so
that we can accommodate their distinct reliability requirements.

In asymmetric MDC, there are $2^{\paths}\!-\!1$ \emph{independent} source sequences, denoted by $U_i,i \in [1:2^{\paths}-1]$ with information rates $R_i, i\in[1:2^{\paths}-1]$. 
We assume that the source sequences are ordered with decreasing importance (i.e., $U_{1}$ is the most important and $U_{2^{\paths}-1}$ the least important). The source sequences are mapped into $\paths$ distinct descriptions through an encoder. At the destination, there are $2^{\paths}\!-\!1$ decoders operating, each connected to a \emph{non‑empty} subset of the descriptions. Based on which descriptions each decoder can access, the decoders are assigned with ordered levels: there are $2^{\paths}-1$ levels. The decoder at level $h \in [1:2^{\paths}-1]$ decodes the $h$ most important source sequences $U_i, i \in [1:h]$.  

\begin{figure}
    \centering   \includegraphics[width=.75\columnwidth]{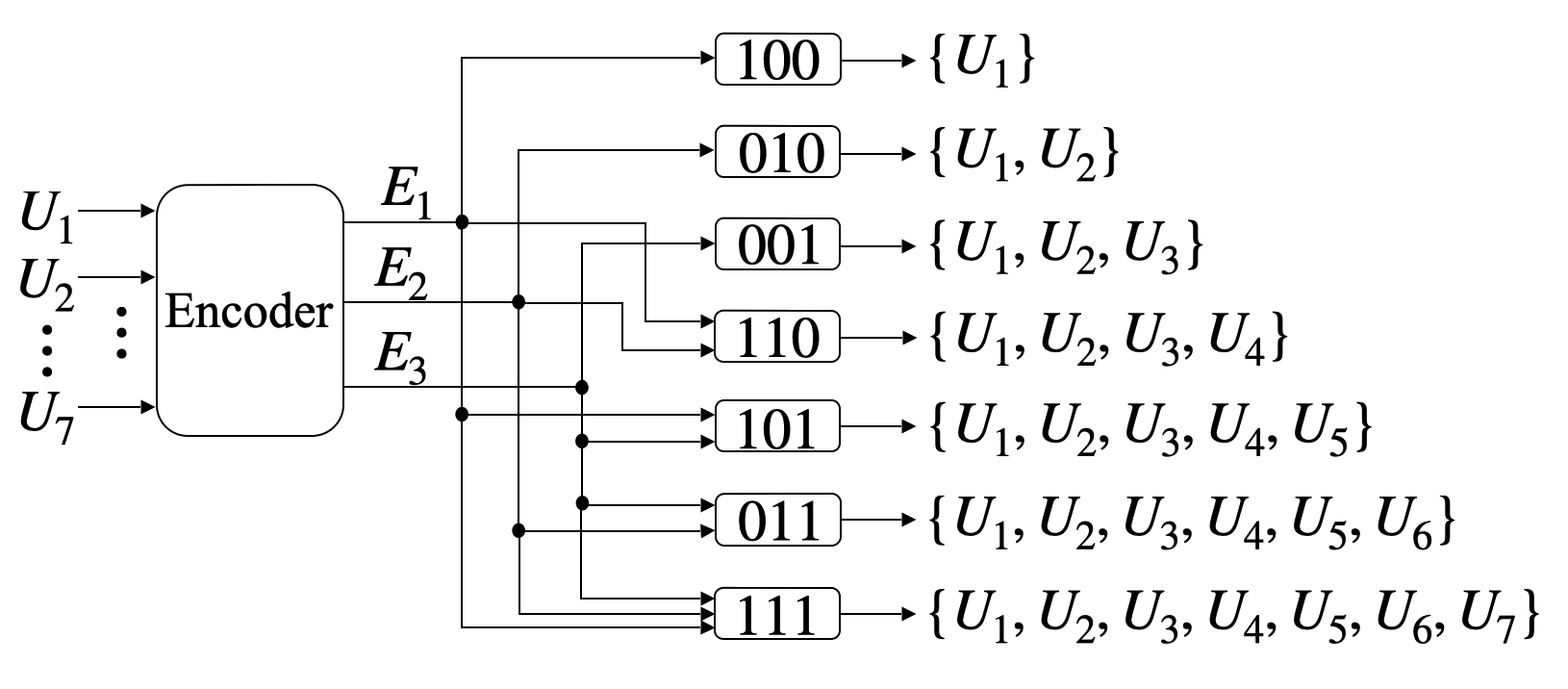}
    \caption{$7$-level asymmetric MDC setting.}
    \label{fig:asymmetric_setting}
    \vspace{-0.15in}
\end{figure}

\noindent\emph{Example 1.} Consider the network in Fig.~\ref{fig:asymmetric_setting} for $\paths =3$. The encoder encodes $2^3-1=7$ source sequences $U_{i},i \in [1:7]$ into $3$ descriptions $E_{i}, i \in [1:3]$. 
The decoder at level~$h \!\in\! [1:7]$ recovers the $h$ most important sequences, e.g., the decoder at level $3$ has access to the third description $001$ and it recovers~$U_{i},i\!\in\![1:3]$.  \hfill \qed

\section{Problem Formulation}\label{sec:Problem_Formulation}
\medskip
In this section, we formally define our problem and present our MDC setting. In particular, we discuss how MDC can be leveraged to design a resilient transmission scheme.

\noindent\textbf{Applying MDC.}
Consider a network with $\paths$ \hbox{edge‑disjoint} paths, denoted by $p_{[1:\paths]}$, each suffering an independent blockage event. As discussed in Section~\ref{sec:Background}, we create $\paths$ packets (i.e., descriptions) by encoding the $2^{\paths} - 1$ source sequences. 
We then transmit these $\paths$ packets, one over each edge-disjoint path. A \emph{blockage pattern} is represented by a binary vector $b \in \{0,1\}^{\paths} \setminus \{\mathbf{0}\}$. In particular, if path $p_i \in p_{[1:\paths]}$ is blocked, then $b(i)=0$ where $b(i)$ denotes the $i$-th element of $b$, and $b(i)=1$ otherwise. 
Since paths fail independently, there are exactly $2^{\paths} - 1$ non‑trivial patterns and the probability of a specific blockage pattern $b$ occurring is defined next. 
\begin{definition}[Blockage Pattern Probability]\label{def:prob_blockage_pattern} 
	Let $\mathcal{B}$ denote the set of all non-zero blockage patterns in the network $($i.e., \hbox{$\mathcal{B} = \{0,1\}^{\paths} \setminus \{\mathbf{0}\})$}, and let $P(b)$ denote the probability of a specific blockage pattern $b \in \mathcal{B}$ occurring. Then,
	\begin{equation}
		P(b) = \prod_{i=1}^{\paths} q_i^{1-b(i)}(1-q_i)^{b(i)},
	\end{equation}
	where $q_i$ denotes the blockage probability of path \hbox{$p_i \in p_{[1:\paths]}$}. 
\end{definition}
A path being unblocked indicates that the description that is sent through that path is received. Thus, every blockage pattern can correspond to one decoder in the MDC setting as illustrated in Fig.~\ref{fig:asymmetric_setting}, and the binary vector representation shows which descriptions the corresponding decoder has access to. That is, when pattern $b$ occurs, the corresponding decoder receives the descriptions indexed by the $1$‑entries of $b$ and can decode the source sequences up to its level.
The complexity of MDC increases as the number of \hbox{edge-disjoint} paths $\paths$ increases ($2^{\paths}-1$ blockage patterns are considered). We build on our recent work~\cite{MineISIT2025} and decrease the design and operational complexity of MDC by dividing the blockage patterns into $m$ different groups $\group_i, i\in[1:m]$. 
Then, we encode $m$ source sequences $U_i,i\in[1:m]$ based on the formed groups instead of individual blockage patterns. We create the groups such that $\bigcup_{i=1}^m \group_i \subseteq \mathcal{B}$ i.e., some blockage patterns might not appear in any of the groups. This is motivated by scenarios where some blockage patterns have a small probability of occurrence (see Definition~\ref{def:prob_blockage_pattern}). In such cases, we can indeed create the groups by ignoring the blockage patterns that occur with a probability that is lower than some specified threshold, which can be selected by an application of interest. This can allow us to achieve higher rates for the source sequences (see Section~\ref{sec:3_paths}).

\begin{remark}\label{remark:2_groups}
While there are $(m+1)^{2^\paths - 1}$ possible ways to assign blockage patterns into $m$ groups, many configurations are redundant for capacity region analysis. Nonetheless, characterizing the capacity region remains challenging for arbitrary values of $m$ and $\paths$. As a first step, in this paper we focus on the practically relevant case of $m = 2$ groups.
\end{remark}
As noted in Remark~\ref{remark:2_groups}, we partition the patterns into two groups, namely $\group_1$ and $\group_2$, corresponding to two independent source sequences: $U_1$ (high priority) and $U_2$ (low priority), with information rates $R_1$ and $R_2$, respectively. We require to decode \textit{at least} $U_1$ if \textit{any} pattern from $\group_1$ occurs and to decode both $U_1$ and $U_2$ if \textit{any} pattern from $\group_2$ occurs. This two-level coding structure is illustrated in Fig.~\ref{fig:2_level_mdc}.
	\begin{figure}
		\centering  \includegraphics[width=0.65\columnwidth]{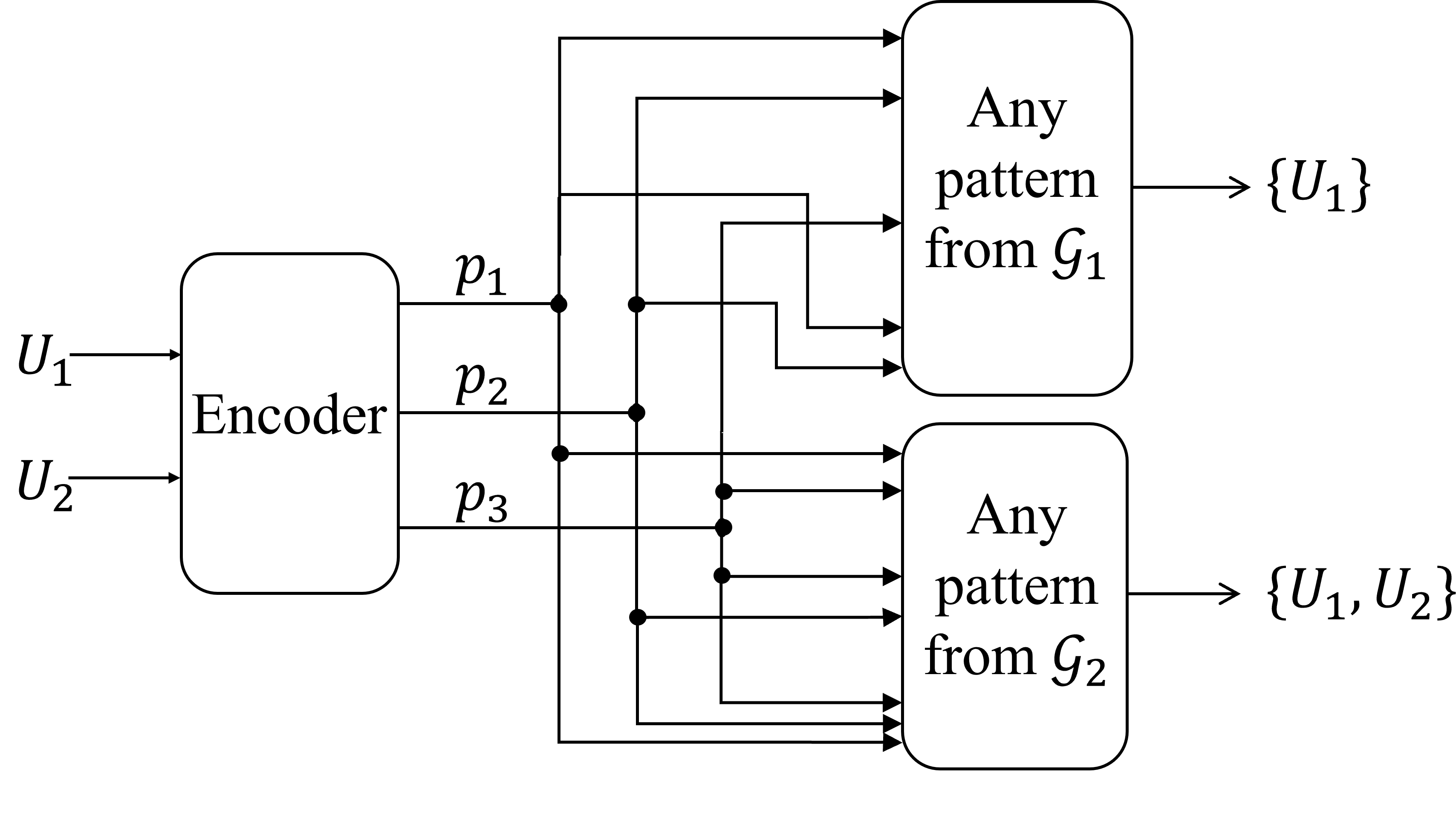}
		\caption{2-level MDC  with $m=2$ groups, $\paths=3$ paths.}
		\label{fig:2_level_mdc}
		\end{figure}
        
For any blockage pattern $b \in \group_1 \cup \group_2$, we let $\mathcal{S}(b)$ denote the set of unblocked paths in pattern $b$. For each group \hbox{$\group_i, i \in [1:2]$}, we let
\begin{equation} \label{eq:min_path Def}
	\kappa_i = |\mathcal{S}(b^\star_i)|, \qquad b^\star_i = \argmin_{b \; \in \; \group_i} |\mathcal{S}(b)|,
\end{equation}
We next define the probability of successfully decoding each source sequence.
\begin{definition}[Probability of Decoding]\label{def:prob_decoding}
	For the group \hbox{$\group_i, i \in [1:2]$}, the probability of decoding the source sequence $U_i$ is denoted by $P(U_i)$ and given by
	\begin{equation}
		P(U_i) = \sum_{j = i}^2 \; \sum_{b \in  \group_j} P(b),
	\end{equation}
	where $P(b)$ is defined in Definition~\ref{def:prob_blockage_pattern}. 
\end{definition}

We note that the priority levels of the source sequences as well as their QoS requirements are determined by the application of interest. However, the capacity regions (see Section~\ref{sec:3_paths}) and the reliability that we can offer for each source sequence (i.e., the probability of decoding that source sequence) are affected by how the groups are designed (i.e., which patterns are included in each group). 
We consider groups $\group_i,i\in\{1,2\}$, that satisfy the two following assumptions:

\noindent {\bf A1:} $\kappa_1 \leq \kappa_2$. We note that this is a reasonable assumption. As we will discuss in Section~\ref{sec:3_paths}, the value of $\kappa_i,i\in \{1,2\}$ constrains the capacity region, which constrains the information rates of the source sequences. For example, $\kappa_1$ constrains the value of $R_1$, while $\kappa_2$ constrains $R_1+R_2$ (see Section~\ref{sec:3_paths} for a detailed analysis). Thus, it is reasonable to create the groups such that $\kappa_1 \leq \kappa_2$.

\noindent {\bf{A2:}} $\mathcal{S}\left( b_2\right) \not \subseteq \mathcal{S}\left( b_1 \right)$ for $b_2 \in \group_2$ and $b_1 \in \group_1$. This assumption is also reasonable since it ensures that more unblocked paths do not decrease the set of source sequences that one can reliably decode. For the same reason if $b= 111 \in \group_1 \cup \group_2$, and $\group_2$ is not empty, then we assume that $b= 111 \in \group_2$. 


	\section{Reduction to Combination Networks}
    \label{sec:reduction}
In~\cite{bidokhti2016capacity}, the authors considered a scenario where a source needs to encode and transmit a {\em common} message and a {\em private} message toward a set of users. In particular, a subset of the users (referred to as public receivers) require only the common message, whereas the remaining users (referred to as private receivers) demand both the common and the private messages. The authors of~\cite{bidokhti2016capacity} focused on a scenario, where the communication network has a special structure, namely it is a {\em combination network}~\cite{ngai2004network}.
\begin{definition}[Combination Network~\cite{ngai2004network}]
A combination network is a three-layered network with a single source and multiple destinations. 
It consists of a source node in the first layer, a set of $K$ intermediate nodes in the second layer, and a set of $D$ destination nodes in the third layer. 
The source is connected to all the intermediate nodes and each intermediate node is connected to a subset of the destinations. Each network link has the same capacity $\capacity$. 
An example of a combination network with $K=3$ and $D=4$ is shown in Fig.~\ref{fig:example_network1}.
\end{definition}
In~\cite{bidokhti2016capacity}, the main focus was on designing high-performing encoding schemes and characterizing the maximum rates of communication that can be achieved over combination networks with public and private receivers. The authors characterized the ultimate rates of communication when the network has either two~\cite[Theorem 3]{bidokhti2016capacity} or three~\cite[Theorem 4]{bidokhti2016capacity} public receivers, and an arbitrary number of private receivers.

We now show that the problem of designing high-performing MDC schemes with two priority levels can be {\em reduced} to the problem of designing high-performing encoding schemes over combination networks with public and private receivers. In particular, each instance of $(\group_1,\group_2)$ can be modeled with a combination network constructed as follows:
\begin{enumerate}
\item Node $0$ of the network is the source node;
\item The $\edgedisjoint$ edge-disjoint paths $p_{[1:\edgedisjoint]}$ in the network are the intermediate nodes; 
\item The $|\group_1 \cup \group_2|$ blockage patterns are the destinations;
\item Node $0$ is connected to all the $\edgedisjoint$ intermediate nodes. An intermediate node $p_i, i \in [1:\edgedisjoint]$ is connected to a destination node $b \in \group_1 \cup \group_2$ if and only if path $p_i$ is unblocked in the blockage pattern $b$. The destination nodes in $\group_1$ are the public receivers (i.e., $U_1$ is the common message) and the destination nodes in $\group_2$ are the private receivers (i.e., $U_2$ is the private message).
\end{enumerate}
\noindent \emph{Example~2.} Consider a network with $\edgedisjoint = 3$. Let $\group_1 = \{100,110\}$ and $\group_2 = \{011,101\}$. This instance of $(\group_1,\group_2)$ can be modeled by the combination network in Fig.~\ref{fig:example_network1}. \hfill \qed

With the above construction, characterizing the maximum rates of communication that can be achieved by an MDC scheme with two priority levels becomes equivalent to the problem of characterizing the maximum flow of information that can be achieved over a combination network.
We will therefore leverage the result in~\cite{bidokhti2016capacity} in the next section.

\begin{figure}
    \centering  \includegraphics[width=0.45\columnwidth]{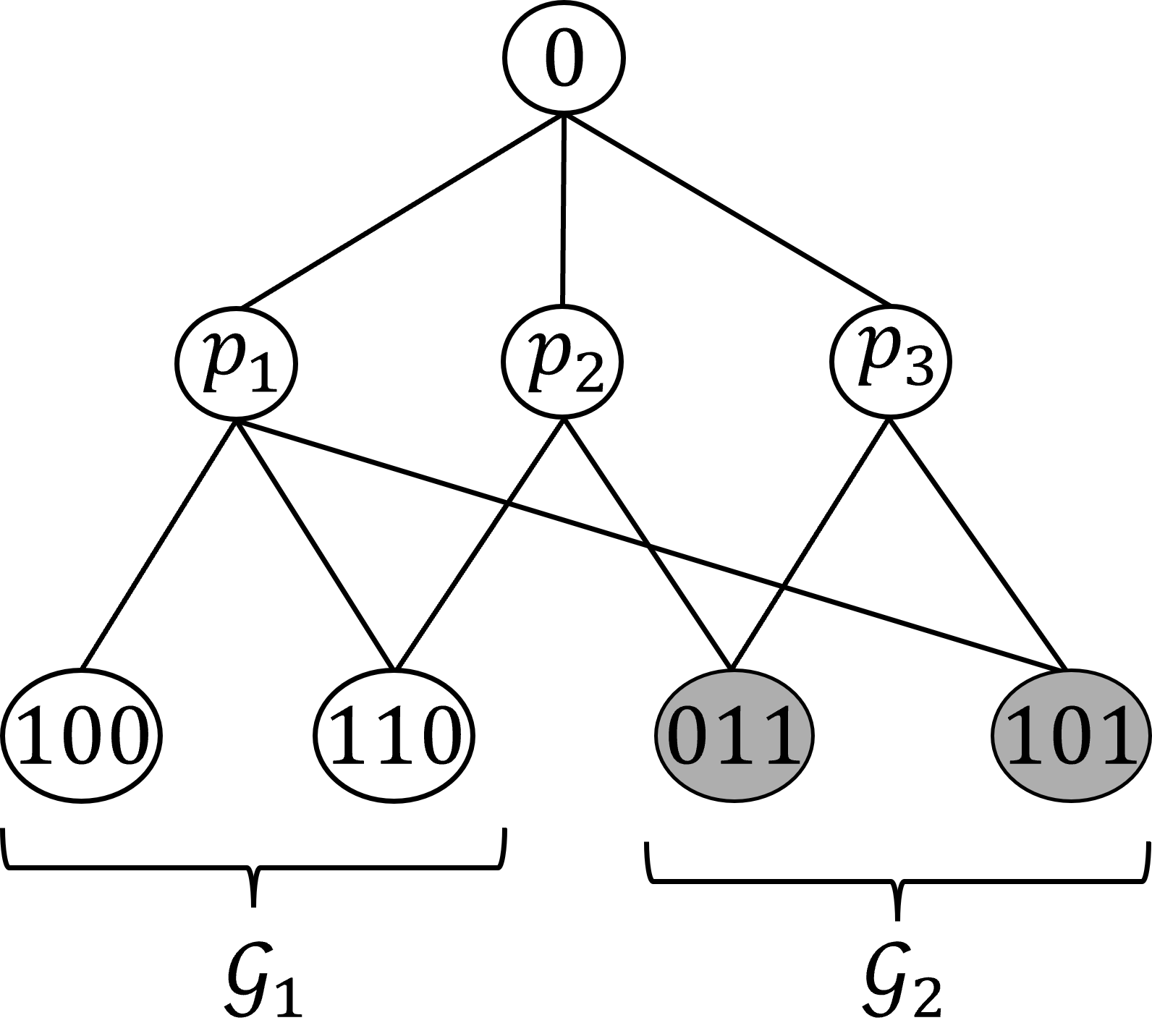}
    \caption{An example of combination network with $K=3$ and $D=4$ that can be used to model the instance $(\group_1,\group_2) = (\{100,110\},\{011,101\})$.}
    \label{fig:example_network1}
    \vspace{-0.15in}
\end{figure}

\section{Capacity Region for $3$ paths}\label{sec:3_paths}
In this section, we present the capacity region characterization for our problem introduced in Section~\ref{sec:Problem_Formulation}. In particular, we present the complete characterization when there are $\edgedisjoint=3$ \hbox{edge-disjoint} paths in the network as illustrated in Fig.~\ref{fig:2_level_mdc}. 
We note that the structure of the groups $\group_1$ and $\group_2$ affects the design of the coding scheme; as such, the capacity region will depend on $\group_1$ and $\group_2$.
In particular, when $\group_1 \cup \group_2 = \mathcal{B}$, i.e., all the $2^\edgedisjoint - 1 = 7$ blockage patterns are considered, then we have recently characterized the capacity region in~\cite[Theorem~2]{MineISIT2025}. Therefore, we here focus on the case where $\group_1 \cup \group_2 \subset \mathcal{B}$.

We start with the following proposition, which considers two `trivial' cases, i.e., either $\group_1$ is empty or $\group_2$ is empty.
\begin{prop}
\label{prop:Trivial}
Consider a network with $\paths=3$ \hbox{edge-disjoint} paths, each with capacity $\capacity$. Divide the blockage patterns into two groups $\group_1$ and $\group_2$ such that \hbox{$\group_1 \cup \group_2 \subset \mathcal{B}$}. Consider two independent source sequences $U_1$ and $U_2$ with information rates $R_1$ and $R_2$, respectively. Then, using $\kappa_i, i \in \{1,2\}$ in~\eqref{eq:min_path Def}, we have the following capacity region characterizations:
\begin{enumerate}
\item If $|\group_2| = 0$, then
\begin{equation}
\label{eq:G2Empty}
R_1 = \kappa_1 \capacity.
\end{equation}
\item If $|\group_1| = 0$, then
\begin{equation}
\label{eq:G1Empty}
R_1+ R_2  = \kappa_2 \capacity.
\end{equation}
\end{enumerate}
\end{prop}
\begin{proof}
We consider the two different cases separately.
\begin{enumerate}
\item {\bf Case $|\group_2| =0$.} All the blockage patterns considered are in $\group_1$ and we require that we decode $U_1$ (with rate $R_1$) if any pattern from $\group_1$ occurs. 
Since we consider edge-disjoint paths with equal capacity $\capacity$ and by the definition of $\kappa_1$ in~\eqref{eq:min_path Def}, we have that $R_1 \leq \kappa_1 \capacity$. As we have shown in~\cite[Theorem 1]{MineISIT2025}, this rate is achievable by using a $(3,\kappa_1)$ erasure code. This proves~\eqref{eq:G2Empty}.
\item {\bf Case $|\group_1| =0$.} All the blockage patterns considered are in $\group_2$ and we require that we decode both $U_1$ and $U_2$ (with sum rate $R_1+R_2$) if any pattern from $\group_2$ occurs. 
Moreover, $R_1+R_2$ cannot exceed the sum of the capacities of the unblocked paths in any pattern from $\group_2$, i.e., we need $R_1+R_2 \leq \kappa_2 \capacity$.
As we have shown in~\cite[Theorem 1]{MineISIT2025}, this rate is achievable by using a $(3,\kappa_2)$ erasure code. This proves~\eqref{eq:G1Empty}.
\end{enumerate}
This concludes the proof of Proposition~\ref{prop:Trivial}.
\end{proof}
We now focus on the case where both $\group_1$ and $\group_2$ are non-empty. We start by considering the cases $|\group_1|=2$ and $|\group_1|=3$ for which, as discussed in Section~\ref{sec:reduction}, we can readily apply the results in~\cite[Theorem 3]{bidokhti2016capacity} and~\cite[Theorem 4]{bidokhti2016capacity}.
\begin{prop}
\label{prop:Shirin}
Consider a network with $\paths=3$ \hbox{edge-disjoint} paths, each with capacity $\capacity$. Divide the blockage patterns into two groups $\group_1$ and $\group_2$ such that \hbox{$\group_1 \cup \group_2 \subset \mathcal{B}$}. Consider two independent source sequences $U_1$ and $U_2$ with information rates $R_1$ and $R_2$, respectively. Then, we have the following capacity region characterizations:
\begin{enumerate}
\item If $|\group_2| = 2$, then any achievable rate pair $(R_1,R_2)$ lies in the rate region of~\cite[Proposition 1]{bidokhti2016capacity}.
\item If $|\group_2| = 3$, then any achievable rate pair $(R_1,R_2)$ lies in the rate region of~\cite[Theorem 1]{bidokhti2016capacity}.
\end{enumerate}
\end{prop}
To complete the characterization of the capacity region when there are $\edgedisjoint=3$ \hbox{edge-disjoint} paths, we still need to consider two cases, namely $|\group_1|=1$ and $|\group_1|>3$. We note that when $|\group_1| > 3$, it necessarily follows that $\kappa_1 = 1$ and $\kappa_2 \geq 2$. These cases are addressed in the following theorem. 
\begin{theorem}\label{theorem:main}
Consider a network with $\paths=3$ \hbox{edge-disjoint} paths, each with capacity $\capacity$. Divide the blockage patterns into two groups $\group_1$ and $\group_2$ such that \hbox{$\group_1 \cup \group_2 \subset \mathcal{B}$}. Consider two independent source sequences $U_1$ and $U_2$ with information rates $R_1$ and $R_2$, respectively. 
Let $\kappa_i, i \in \{1,2\}$ and ${b}_i^\star, i \in \{1,2\}$ as in~\eqref{eq:min_path Def} and let
\begin{align*}
\mathcal{C}_1 = \{&\text{There exist at least $\kappa_2$ patterns in $\group_1$, each with} \notag 
\\& \text{exactly one unblocked path, such that the union of} \notag
\\& \text{their unblocked paths equals $\mathcal{S}({b}_2^\star)$}\}.
\end{align*}
Then, we have the following capacity region characterizations:
\begin{enumerate}
\item If $\mathcal{C}_1$ is satisfied, then 
\begin{equation}\label{eq:opt_RateRegion}
R_1+\frac{R_2}{\kappa_2} = \capacity.
\end{equation}
\item If $\mathcal{C}_1$ is not satisfied and if one of the following holds: (i) $|\group_1|=1$; or (ii) $|\group_1|>3$ with $\kappa_2=2$,  then
\begin{subequations}
\label{eq:CapSeconRegime}
\begin{align}\label{eq:cut_sets}
R_1 &= \kappa_1 \capacity,
\\
R_1 + R_2 &=\kappa_2 \capacity.
\end{align}
\end{subequations}
%
%
\item If $\mathcal{C}_1$ is not satisfied and $|\group_1|>3$ with $\kappa_2=3$, then
\begin{subequations}
\label{eq:opt_RateRegion2}
\begin{align}
R_1 & =  \kappa_1 \capacity,\label{eq:opt_RateRegion2_cutset}
\\
2R_1+R_2 & = 3\capacity.\label{eq:opt_RateRegion2_new}
\end{align}
\end{subequations}

\end{enumerate}
\end{theorem}
\begin{proof}
We consider the three different cases separately and, in all of them, we use the fact that $\kappa_1\leq \kappa_2 \leq 3$. 
\begin{enumerate}
\item When $\mathcal{C}_1$ holds, we show that $R_1 +\frac{R_2}{\kappa_2}\leq \capacity$ in Appendix~\ref{app:cond_optimality}. 
As we have shown in~\cite[Theorem~1]{MineISIT2025}, this bound can be achieved by encoding $U_1$ and $U_2$ separately with an erasure code $(\paths,\kappa_i)$. This proves~\eqref{eq:opt_RateRegion}. 
\item When $\mathcal{C}_1$ does not hold, we derive an upper bound as follows. We require to decode at least $U_1$ (with rate $R_1$) if any pattern from $\group_1$ occurs, and to decode both $U_1$ and $U_2$ (with sum rate $R_1+R_2$) if any pattern from $\group_2$ occurs. Thus, $R_1$ (respectively, $R_1+R_2$) cannot exceed the sum of the capacities of the unblocked paths in any pattern from $\group_1$ (respectively, $\group_2$). Since we consider edge-disjoint paths with equal capacity $\capacity$ and by the definition of $\kappa_i,i\in\{1,2\}$ in~\eqref{eq:min_path Def}, we have that 
\begin{subequations}
\label{eq:UBG1CardOne}
\begin{align}
R_1 & \leq  \kappa_1 \capacity, \label{eq:R1CS}
\\
R_1 + R_2 & \leq \kappa_2 \capacity. \label{eq:R1plusR2CS}
\end{align}
\end{subequations}
To show achievability, we distinguish two subcases, namely $|\group_1|=1$, and $|\group_1|>3$ with $\kappa_2=2$.

\noindent  $\bullet \ |\group_1|=1$. 
To show achievability of the above bounds, we consider two different cases:
\begin{enumerate}
\item $\kappa_1 = \kappa_2$. For this case, the bound in~\eqref{eq:R1CS} is redundant and we can achieve~\eqref{eq:R1plusR2CS} by encoding both $U_1$ and $U_2$ with $(3,\kappa_1)$ erasure codes.
\item $\kappa_1 \neq \kappa_2$. We show the achievability of the three corner points of~\eqref{eq:UBG1CardOne}, namely $(R_1,R_2)= (0,\kappa_2\capacity)$, $(R_1,R_2)= (\kappa_1 \capacity,(\kappa_2-\kappa_1) \capacity)$, and $(R_1,R_2)= (\kappa_1 \capacity,0)$.
The rate pair $(0,\kappa_2 \capacity)$ can be achieved by encoding $U_2$ with a $(3,\kappa_2)$ code, and the rate pair $(\kappa_1 \capacity,0)$ can be achieved by encoding $U_1$ with a $(3,\kappa_1)$ code.
Finally, to achieve the rate pair $(\kappa_1 \capacity,(\kappa_2-\kappa_1) \capacity)$, we distinguish two cases based on the value of $\kappa_2$: (i) $\kappa_2=2$, and (ii) $\kappa_2=3$. 
When $\kappa_2=2$, we have $\kappa_1 = 1$, and the target rate pair is $(\capacity,\capacity)$. This can be achieved by: (i) transmitting $U_1$ through the unblocked path in the pattern in $\group_1$; (ii) transmitting $U_2$ through any of the remaining two paths; and (iii) transmitting $U_1 + U_2$ through the remaining path. 
When $\kappa_2=3$, we have that $\kappa_1 \in \{1,2\}$. The rate pair $(\kappa_1 \capacity,(3-\kappa_1) \capacity$ can be achieved by: (i) splitting $U_1$ into $\kappa_1$ equal parts, and transmitting them over the $\kappa_1$ unblocked paths in the pattern in $\group_1$; and (ii) splitting $U_2$ into $3-\kappa_1$ equal parts, and transmitting them over the remaining $3-\kappa_1$ paths.
\end{enumerate}

\noindent $\bullet \ |\group_1|>3$ with $\kappa_2=2$.
In Appendix~\ref{app:achievability}, we show that there exist simple coding schemes that achieve the bounds in~\eqref{eq:UBG1CardOne}. This concludes the proof of~\eqref{eq:CapSeconRegime}.


\item When $\mathcal{C}_1$ does not hold and $|\group_1|>3$ with $\kappa_2=3$, we provide the detailed proof in Appendix~\ref{app:final_rate_region}.

\end{enumerate}
This concludes the proof of Theorem~\ref{theorem:main}.
\end{proof}
We present examples for each case in Theorem~\ref{theorem:main} to illustrate the conditions and the corresponding capacity regions.

\noindent \emph{Example~3.} This example illustrates Case 1 in Theorem~\ref{theorem:main}. Consider a network with $\paths = 3$. Let $\group_1 = \{100,010,101,011\}$ and $\group_2 = \{110,111\}$. From~\eqref{eq:min_path Def}, we have $\kappa_2=2$, $b^\star_2=110$, and $\mathcal{S}(b^\star_2) = \{p_1, p_2\}$ in this example. The patterns $b_{1,1}= 100$ and $b_{1,2}=010$ in $\group_1$ have exactly one unblocked path each, and they satisfy $\mathcal{S}(b_{1,1})\cup \mathcal{S}(b_{1,2})=\mathcal{S}(b^\star_2)$. Thus, $\mathcal{C}_1$ is satisfied and the capacity region is characterized by~\eqref{eq:opt_RateRegion}.  \hfill \qed

\noindent \emph{Example~4.} This example illustrates Case 2 in Theorem~\ref{theorem:main} for $|\group_1|>3$ with $\kappa_2=2$. Consider a network with $\paths = 3$. Let $\group_1 = \{100,010,110,011\}$ and $\group_2 = \{101,111\}$. We have $|\group_1|=4>3$, and as defined in~\eqref{eq:min_path Def}, we have $\kappa_2=2$, $b^\star_2=101$, and $\mathcal{S}(b^\star_2) = \{p_1, p_3\}$. The patterns $b_{1,1}= 100$ and $b_{1,2}=010$ in $\group_1$ have exactly one unblocked path each, but $\mathcal{S}(b_{1,1})\cup \mathcal{S}(b_{1,2})=\{p_1,p_2\}\neq \mathcal{S}(b^\star_2)$. Thus, $\mathcal{C}_1$ is not satisfied, and the capacity region is characterized by~\eqref{eq:CapSeconRegime}. \hfill \qed

\noindent \emph{Example 5.} This example illustrates Case 3 in Theorem~\ref{theorem:main}. Consider a network with $\paths = 3$. Let $\group_1 = \{100,010,110,011\}$ and $\group_2 = \{111\}$. In this example, we have $|\group_1|=4>3$, and as defined in~\eqref{eq:min_path Def}, we have $\kappa_2=3$, $b^\star_2=111$, and $\mathcal{S}(b^\star_2) = \{p_1,p_2, p_3\}$. The condition $\mathcal{C}_1$ requires that there exist $\kappa_2=3$ patterns in $\group_1$, each with exactly one unblocked path. However, there are only two such patterns $b_1= 100$ and $b_2=010$ in $\group_1$. Thus, $\mathcal{C}_1$ is not satisfied, and the capacity region is characterized by~\eqref{eq:opt_RateRegion2}. \hfill \qed


\section{Conclusions and Future Work} \label{sec:conclusion}
In this work, we have deployed an MDC scheme over a communication network to offer resilience against link blockages, while providing distinct reliability guarantees based on the priority level of data streams. We reduced the design of high-performing MDC schemes to an equivalent problem over combination networks and characterized the capacity region for two priority levels and three edge-disjoint paths. Building on these results, our ongoing investigation is to derive the capacity region for a larger number of paths. Our recent results in~\cite{MineISIT2025} presented the capacity region for an arbitrary number of paths when all blockage patterns are considered, and the results in~\cite{bidokhti2016capacity} present the capacity region for an arbitrary number of paths when $|\group_1| = 2$ or $|\group_1|=3$. Moreover, Case~1 in Theorem~\ref{theorem:main}, in fact, characterizes the capacity region for an arbitrary number of paths when condition $\mathcal{C}_1$ holds. These initial results provide a first step towards understanding how to extend our analysis to a larger number of paths.    

\bibliographystyle{IEEEtran}
\bibliography{IEEEabrv,bibliography}

\appendices
\section{Proof of Upper Bound in Case~1 of Theorem~\ref{theorem:main}}\label{app:cond_optimality}
We let $X_i^n$ and $Y_i^n$ be the set of packets sent and received over path $p_i, i \in [1:\paths]$, respectively, in $n$ network uses. 
Consider a blockage pattern $b^\star_2$ in $\group_2$ that satisfies the condition $\mathcal{C}_1$ with the number of unblocked paths equal to $\kappa_2$. 
Without loss of generality, we assume that the first $\kappa_2$ paths are unblocked.
We have that
\begin{align}
n(R_1+R_2) &= H(U_1,U_2) \notag
\\& \stackrel{{\rm{(a)}}}{=} I\left(U_1,U_2;Y^n_1,\ldots,Y^n_{\kappa_2}\right) \notag
\\& \quad +H\left(U_1,U_2\mid Y^n_1,\ldots,Y^n_{\kappa_2}\right) \notag 
\\& \stackrel{{\rm{(b)}}}{=} I\left(U_1,U_2;Y^n_1,\ldots,Y^n_{\kappa_2}\right) \notag
\\& \stackrel{{\rm{(c)}}}{\leq} I\left(X^n_1,\ldots,X^n_{\kappa_2};Y^n_1,\ldots,Y^n_{\kappa_2}\right) \notag
\\& \stackrel{{\rm{(d)}}}{=} H\left(X^n_1,\ldots,X^n_{\kappa_2}\right) \notag
\\& \stackrel{{\rm{(e)}}}{=} H\left(X^n_1\right) +\sum_{i=2}^{\kappa_2} H\left(X^n_i\mid X_1^n,\ldots,X^n_{i-1}\right) \notag
\\& \stackrel{{\rm{(f)}}}{\leq} \sum_{i=1}^n H\left(X_{1,i}\right) +\sum_{i=2}^{\kappa_2} H\left(X^n_i\mid X_1^n,\ldots,X^n_{i-1}\right) \notag
\\& \stackrel{{\rm{(g)}}}{\leq} n \capacity +\sum_{i=2}^{\kappa_2} H\left(X^n_i\mid X_1^n,\ldots,X^n_{i-1}\right) \notag
\\& \stackrel{{\rm{(h)}}}{=} n \capacity +\sum_{i=2}^{\kappa_2} H\left(X^n_i\mid X_1^n,\ldots,X^n_{i-1},U_1\right) \notag
\\& \stackrel{{\rm{(f)}}}{\leq}  n \capacity +\sum_{i=2}^{\kappa_2} H\left(X^n_i\mid U_1\right) \notag
\\& \stackrel{{\rm{(a)}}}{=} n \capacity +\sum_{i=2}^{\kappa_2} \left( H\left(X^n_i\right)-I\left(X^n_i;U_1\right) \right) \notag
\\& \stackrel{{\rm{(a)}}}{=} n \capacity +\sum_{i=2}^{\kappa_2} \left ( H\left(X^n_i\right)\!-\!H\left(U_1\right)\!+\!H\left(U_1\mid X_i^n\right) \right ) \notag
\\& \stackrel{{\rm{(h)}}}{=} n \capacity +\sum_{i=2}^{\kappa_2} \left ( H\left(X^n_i\right)-H\left(U_1\right) \right ) \notag
\\& = n \capacity +\sum_{i=2}^{\kappa_2} \left ( H\left(X^n_i\right)-nR_1 \right ) \notag
\\& \stackrel{{\rm{(f)+(g)}}}{\leq} n \capacity +\sum_{i=2}^{\kappa_2} \left ( n \capacity-nR_1 \right ) \notag
\\& = n \capacity + n \left( \kappa_2-1 \right ) \left(\capacity - R_1\right) \notag
\\& = n\kappa_2 \capacity-n(\kappa_2-1)R_1,
\label{eq:FinalStepOB}
\end{align}
where the labeled (in)equalities follow from:
$\rm{(a)}$ the mutual information definition;
$\rm{(b)}$ the fact that $U_1$ and $U_2$ are both fully determined when $\left(Y^n_1,\ldots,Y^n_{\kappa_2}\right)$ are received as we consider a blockage pattern in $\group_2$;
$\rm{(c)}$ the data processing inequality, the chain rule of the mutual information, and the zero-error capacity channels assumption;
$\rm{(d)}$ the \hbox{zero-error} capacity channels assumption;
$\rm{(e)}$ the chain rule of the entropy;
$\rm{(f)}$ the fact that conditioning does not increase entropy;
$\rm{(g)}$ the definition of capacity and the assumption of zero-error capacity channels;
and $\rm{(h)}$ the facts that: (i) the condition $\mathcal{C}_1$ 
(i.e., there exist $\kappa_2$ patterns in $\group_1$ such that these patterns are single-path patterns and the union of their unblocked paths give $\mathcal{S}(b^\star_2)$) holds, 
and (ii) we require $U_1$ (i.e., the most important source sequence) to be decoded for any blockage pattern in $\group_1$, i.e.,  $H\left(U_1\mid X_j^n\right) = 0$ for every $j \in [1:\kappa_2]$.
When we bring the term $n(\kappa_2-1)R_1$ to the \hbox{left-hand} side of the inequality, the outer bound on the sum rate matches the one in~\eqref{eq:opt_RateRegion}. 

\section{Achievability for $|\group_1| > 3$ with $\kappa_2=2$}\label{app:achievability} 

Since there exist $\paths=3$ edge-disjoint paths in the network, there are at most $2^\paths-1 = 7$ possible blockage patterns. 
Given that $\kappa_2 = 2$, then there must exist at least one pattern $b^\star_2 \in \group_2$ in which exactly two of the three paths are unblocked. Moreover, no more than two patterns in which exactly two of the three paths are unblocked can belong to $\group_2$, otherwise this would violate $|\group_1|>3$. However, if $\group_2$ has two patterns in which exactly two of the three paths are unblocked, then all the single-path patterns need to belong to $\group_1$, otherwise $|\group_1|>3$ would again be violated. In this case, condition $\mathcal{C}_1$ would be satisfied and hence, we would fall in the first case of Theorem~\ref{theorem:main}. It therefore follows that we just need to focus on the case where: (i) there exists exactly one pattern in $\group_2$ that has $2$ unblocked paths (otherwise condition $\mathcal{C}_1$ would be satisfied); (ii) $\group_1$ does not contain all the three single-path patterns (otherwise condition $\mathcal{C}_1$ would be satisfied); (iii) $\group_1$ contains exactly four patterns, two of the three single-path patterns and two of the three patterns with two unblocked paths (otherwise $|\group_1|>3$ would be violated); and (iv) the two single-path patterns in $\group_1$ and the pattern with two unblocked paths in $\group_2$ do not satisfy $\mathcal{C}_1$. 
We also note that, since $\kappa_2=2$, $\group_2$ cannot include any pattern with a single unblocked path. However, pattern $111$ can still be included in $\group_2$. 


We now focus on the aforementioned case and show the achievability of the three corner points of~\eqref{eq:UBG1CardOne}, namely $(R_1,R_2)=(0,2\capacity)$, $(R_1,R_2)=(\capacity,\capacity)$, and $(R_1,R_2)=(\capacity,0)$. 
As we have shown in~\cite[Theorem 1]{MineISIT2025}, the rate pair $(0,2\capacity)$ can be achieved by encoding $U_2$ with $(3,2)$ code, and the rate pair $(\capacity,0)$ can be achieved by encoding $U_1$ with $(3,1)$ code. The third rate pair $(\capacity,\capacity)$ can be achieved by: (i) transmitting $U_1$ through each unblocked path that appears in a single-path pattern in $\group_1$ (there are two single-path patterns in $\group_1$ due to the aforementioned reasons); and (ii) transmitting $U_2$ through the remaining path. 

\section{Proof of Theorem~\ref{theorem:main} for $|\group_1| > 3$ with $\kappa_2 = 3$}\label{app:final_rate_region}
We begin with a systematic identification of the group structures that do not satisfy the condition $\mathcal{C}_1$ and for which $|\group_1| > 3$ and $\kappa_2 = 3$.
Since there exist $\paths=3$ edge-disjoint paths in the network, there are at most $2^\paths-1 = 7$ possible blockage patterns. The fact that $\kappa_2=3$ implies that $\group_2$ only contains pattern $111$, and the remaining patterns can only be included in $\group_1$. Moreover, $\group_1$ cannot contain all the three single-path patterns, otherwise condition $\mathcal{C}_1$ would be satisfied. 
We now consider all the cases for which condition $\mathcal{C}_1$ is not satisfied and for which $|\group_1| > 3$ and $\kappa_2 = 3$, by separating them into three scenarios.
In the {\em first} scenario, $\group_1$ includes one single-path pattern and all three patterns with two unblocked paths. In the {\em second} scenario, $\group_1$ includes two single-path patterns and all three patterns that have two unblocked paths. In the {\em third} scenario, $\group_1$ includes two single-path patterns and two of the three patterns with two unblocked paths. We analyze each case separately. We note that for all of these case, $R_1\leq \kappa_1 \capacity$ as we have proved in~\eqref{eq:R1CS}. 




\noindent \textbf{Scenario~1.} 
Without loss of generality, we assume that the \hbox{single-path} pattern in $\group_1$ is $100$. We derive an outer bound on the sum rate using entropy inequalities, following the initial steps outlined in Appendix~\ref{app:cond_optimality},
\begin{align}
n(R_1+R_2) &\leq H\left(X^n_1,X^n_2,X^n_3\right) \notag
\\& \stackrel{{\rm{(a)}}}{=} H\left(X^n_1\right) + H\left(X^n_2,X^n_3 \mid X^n_1\right) \notag
\\& \stackrel{{\rm{(b)}}}{\leq} n \capacity + H\left(X^n_2,X^n_3 \mid X^n_1\right) \notag
\\& \stackrel{{\rm{(c)}}}{=} n \capacity + H\left(X^n_2,X^n_3 \mid X^n_1, U_1\right) \notag
\\& \stackrel{{\rm{(d)}}}{\leq} n \capacity + H\left(X^n_2,X^n_3 \mid U_1\right) \notag
\\& \stackrel{{\rm{(e)}}}{=}  n \capacity + H\left(X^n_2,X^n_3\right)-I\left(X^n_2,X^n_3;U_1\right) \notag
\\& \stackrel{{\rm{(f)}}}{=} n \capacity +H\left(X^n_2,X^n_3\right)- H\left(U_1\right) \notag
\\& \leq 3n \capacity- nR_1, 
\label{eq:FinalStepOB1}
\end{align}
where the labeled (in)equalities follow from:
$\rm{(a)}$ the chain rule of the entropy; $\rm{(b)}$ the definition of capacity and the assumption of zero-error capacity channels; $\rm{(c)}$ the fact that we require $U_1$ (i.e., the most important source sequence) to be decoded for pattern $100$, i.e.,  $H\left(U_1\mid X_1^n\right) = 0$;
$\rm{(d)}$ the fact that conditioning does not increase entropy;
$\rm{(e)}$ the definition of mutual information;
and $\rm{(f)}$ the definition of mutual information and the fact that we require $U_1$ to be decoded for pattern $011$, i.e.,  $H\left(U_1\mid X_2^n, X^n_3\right) = 0$.
When we bring the term $nR_1$ to the \hbox{left-hand} side of the inequality, we obtain the outer bound in~\eqref{eq:opt_RateRegion2_new}. We note that, even if we assumed that the \hbox{single-path} pattern in $\group_1$ is $100$ and applied the chain rule in $\rm{(a)}$ with respect to the entropy of $X^n_1$, the argument is symmetric and applies to any $X_i^n,i\in[1:3]$ since all two-path patterns are included in $\group_1$ in this scenario. 

We next show the achievability of the three corner points of~\eqref{eq:opt_RateRegion2}, namely $(R_1,R_2)=(0,3 \capacity)$, $(R_1,R_2)=(\capacity,\capacity)$, and $(R_1,R_2)=(\capacity,0)$. 
As we have shown in~\cite[Theorem 1]{MineISIT2025}, the rate pair $(0,3\capacity)$ can be achieved by encoding $U_2$ with a $(3,3)$ code. The rate pair $(\capacity,0)$ can be achieved by encoding $U_1$ with a $(3,1)$ code. The third rate pair $(\capacity,\capacity)$ can be achieved by: (i) transmitting $U_1$ through the unblocked path of the \hbox{single-path} pattern in $\group_1$ as well as through one of the other two remaining paths; and (ii) transmitting $U_2$ through the remaining path. 

\noindent \textbf{Scenario 2 and Scenario 3.} In both Scenario~2 and Scenario~3, $\group_1$ includes two single-path patterns. The same steps as in Scenario 1 can be followed to derive the same outer bound in~\eqref{eq:opt_RateRegion2_new}. The rate region has the same corner points as those in Scenario 1. The corner points $(0,3\capacity)$ and $(\capacity,0)$ can be achieved by using the same encoding schemes as those in Scenario 1. The third rate tuple $(\capacity,\capacity)$ can be achieved by: (i) transmitting $U_1$ through the unblocked paths of the two \hbox{single-path} patterns in $\group_1$; and (ii) transmitting $U_2$ through the remaining path.


\end{document}